\documentclass[conference]{IEEEtran}
\usepackage{amsmath,amssymb,amsthm,cite,graphicx,array}
\usepackage{graphicx}
\usepackage{float}
\usepackage{caption}
\usepackage{multicol}
\usepackage{mathrsfs}
\usepackage{amsmath}
\usepackage{amssymb}
\usepackage{amsthm}
\usepackage{multicol}
\usepackage{algorithm}
\usepackage[noend]{algpseudocode}
\usepackage{float}
\usepackage{placeins}
\usepackage{epstopdf}
\usepackage{multirow}
\usepackage{subfigure}
\usepackage{enumitem}
\usepackage[utf8]{inputenc}
\floatstyle{ruled}
\newfloat{algorithm}{tbp}{loa}
\providecommand{\algorithmname}{Algorithm}
\floatname{algorithm}{\protect\algorithmname}
\theoremstyle{remark}
\newtheorem{theorem}{Theorem}

\newtheorem{note}{Note}
\newtheorem{corollary}{Corollary}
\newtheorem{lemma}{Lemma}
\newtheorem{definition}{Definition}

\theoremstyle{remark}
\newtheorem{example}{Example}
\ifCLASSINFOpdf
\else
\fi

\title{On the Broadcast Rate of Index Coding Problems with Symmetric and Consecutive Interference }

\begin{document}

\author{Mahesh~Babu~Vaddi~and~B.~Sundar~Rajan\\ 
 Department of Electrical Communication Engineering, Indian Institute of Science, Bengaluru 560012, KA, India \\ E-mail:~\{vaddi,~bsrajan\}@iisc.ac.in }
 
\maketitle
\begin{abstract}
A single unicast index coding problem (SUICP) with symmetric and consecutive interference (SCI) has $K$ messages and $K$ receivers, the $k$th receiver $R_k$ wanting the $k$th message $x_k$ and having interference $\mathcal{I}_k= \{x_{k-U-m},\dots,x_{k-m-2},x_{k-m-1}\}\cup\{x_{k+m+1}, x_{k+m+2},\dots,x_{k+m+D}\}$ and side-information $\mathcal{K}_k=(\mathcal{I}_k \cup x_k)^c$. In this paper, we derive a lowerbound on the broadcast rate of single unicast index coding problem with symmetric and consecutive interference (SUICP(SCI)). In the SUICP(SCI), if $m=0$, we refer this as single unicast index coding problem with symmetric and neighboring interference (SUICP(SNI)). 
In our previous work\cite{VaR5}, we gave the construction of near-optimal vector linear index codes for SUICP(SNI) with arbitrary $K,D,U$. In this paper, we convert the SUICP(SCI) into SUICP(SNI) and give the construction of near-optimal vector linear index codes for SUICP(SCI) with arbitrary $K,U,D$ and $m$. The constructed codes are independent of field size. The near-optimal vector linear index codes of SUICP(SNI) is a special case of near-optimal vector linear index codes constructed in this paper for SUICP(SCI) with $m=0$. In our previous work\cite{VaR6}, we derived an upperbound on broadcast rate of SUICP(SNI). In this paper, we give an upperbound on the broadcast rate of SUICP(SCI) by using our earlier result on the upperbound on the broadcast rate of SUICP(SNI). We derive the capacity of SUICP(SCI) for some special cases.
\end{abstract}
\section{Introduction and Background}
\label{sec1}
\IEEEPARstart {A}{n} index coding problem (ICP), comprises a transmitter that has a set of $K$ messages, $X=\{ x_0,x_1,\ldots,x_{K-1}\}$, and a set of $M$ receivers, $R=\{ R_0,R_1,\ldots,R_{M-1}\}$. Each receiver, $R_k=(\mathcal{K}_k,\mathcal{W}_k)$, knows a subset of messages, $\mathcal{K}_k \subset X$, called its \textit{side-information}, and wants to know another subset of messages, $\mathcal{W}_k \subseteq \mathcal{K}_k^\mathsf{c}$, called its \textit{Want-set}. The transmitter can take cognizance of the side-information of the receivers and broadcast coded messages, called the index code, over a noiseless channel. Birk and Kol introduced the problem of index coding with side-information in \cite{ISCO}. An index coding problem is single unicast if the demand-sets of the receivers are disjoint and the cardinality of demand-set of every receiver is one \cite{OnH}. The single unicast index coding problems were studied in \cite{YBJK}.



A solution to the ICP may be linear or nonlinear. A solution of the ICP must specify a finite alphabet $\mathcal{A}_P$, and an encoding scheme $\varepsilon:\mathcal{A}^t \rightarrow \mathcal{A}_P$ such that every receiver is able to decode the wanted message from $\varepsilon(x_0,x_1,\ldots,x_{K-1})$ and the known information. The minimum encoding length $l=\lceil log_{2}|\mathcal{A}_P|\rceil$ for messages that are $t$ bit long ($\vert\mathcal{A}\vert=2^t$) is denoted by $\beta_{t}$. The broadcast rate of the ICP is defined \cite{ICVLP} as,
$\beta \triangleq   \inf_{t} \frac{\beta_{t}}{t}.$
For an ICP, $\beta$ is the minimum number of index code symbols required to transmit to satisfy the demands of all the receivers. The capacity $C$ of an ICP is the reciprocal of the broadcast rate $\beta$. 



A single unicast index coding problem with symmetric and consecutive interference (SUICP(SCI)) with equal number of $K$ messages and receivers, is one with each receiver having a total of $U+D<K-2m$ interference, corresponding to the $D~(U \leq D)$ messages after and $U$ messages before its desired message. In this setting, the $k$th receiver $R_{k}$ demands the message $x_{k}$ having the interference
\begin{align}
\label{int1}
\nonumber
{I}_k= &\{x_{k-U-m},\dots,x_{k-m-2},x_{k-m-1}\} \cup  \\& 
~~~~~\{x_{k+m+1}, x_{k+m+2},\dots,x_{k+m+D}\}, 
\end{align}
\noindent
the side-information being 
\begin{align}
\label{sideinfo1}
\nonumber
&\mathcal{K}_k=(\mathcal{I}_k \cup x_k)^c =\{x_{k+m+D+1},x_{k+m+D+2}\ldots x_{k-m-U-1}\} \\& \cup \{x_{k-m},x_{k-m+1}\ldots x_{k-1}\} \cup \{x_{k+1},x_{k+2}\ldots x_{k+m}\}.
\end{align}

The SUICP(SCI) is pictorially represented in Fig. \ref{fig11}. All the subscripts in this paper are to be considered $~\text{\textit{modulo}}~ K$.
\begin{figure*}[ht]
\centering
\includegraphics[scale=0.55]{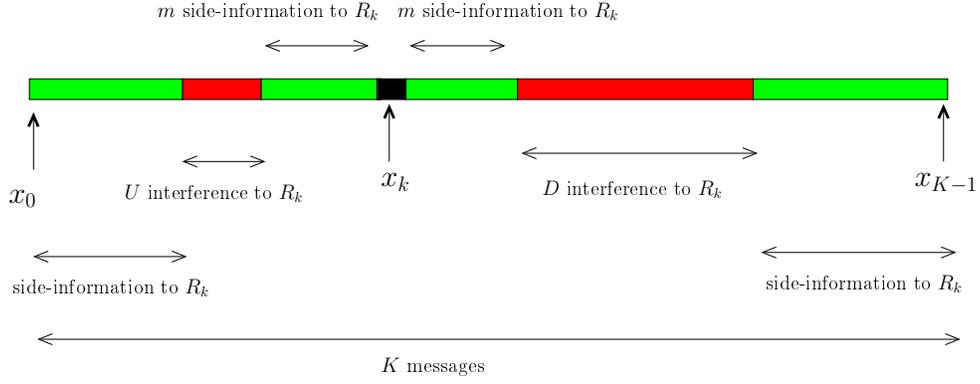}\\
\caption{SUICP(SCI).}
\label{fig11}
~ \\
\hrule
\end{figure*}

A single unicast index coding problem with symmetric neighboring interference (SUICP(SNI)) with equal number of $K$ messages and receivers, the $k$th receiver $R_k$ demands the message $x_k$ having the interference
\begin{equation}
\label{int2}
{I}_k= \{x_{k-U},\dots,x_{k-2},x_{k-1}\}\cup\{x_{k+1}, x_{k+2},\dots,x_{k+D}\}, 
\end{equation}
\noindent
the side-information being 
\begin{align}
\label{sideinfo2}
\mathcal{K}_k=(\mathcal{I}_k \cup x_k)^c=\{x_{k+D+1},x_{k+D+2},\ldots,x_{k-U-1}\}.
\end{align}

Note that the SUICP(SNI) is a special case of SUICP(SCI) with $m=0$.


The existing capacity results for SUICP(SNI) were summarized in \cite{VaR5}. There exists no capacity and optimal index coding results for SUICP(SCI) with $m \neq 0$. The symmetric index coding problems are motivated by topological interference management problems in wireless communication networks \cite{TIM}.
\subsection{AIR matrices}
In \cite{VaR2}, we constructed binary matrices of size $m \times n (m\geq n)$ such that any $n$ adjacent rows of the matrix are linearly independent over every field. We refer these matrices as adjacent row independent (AIR) matrices. The matrix obtained by Algorithm \ref{algo2} is called the $(m,n)$ AIR matrix and it is denoted by $\mathbf{L}_{m\times n}.$ The general form of the $(m,n)$ AIR matrix is shown in Fig. \ref{fig1}. 
\begin{figure*}[ht]
\centering
\includegraphics[scale=0.5]{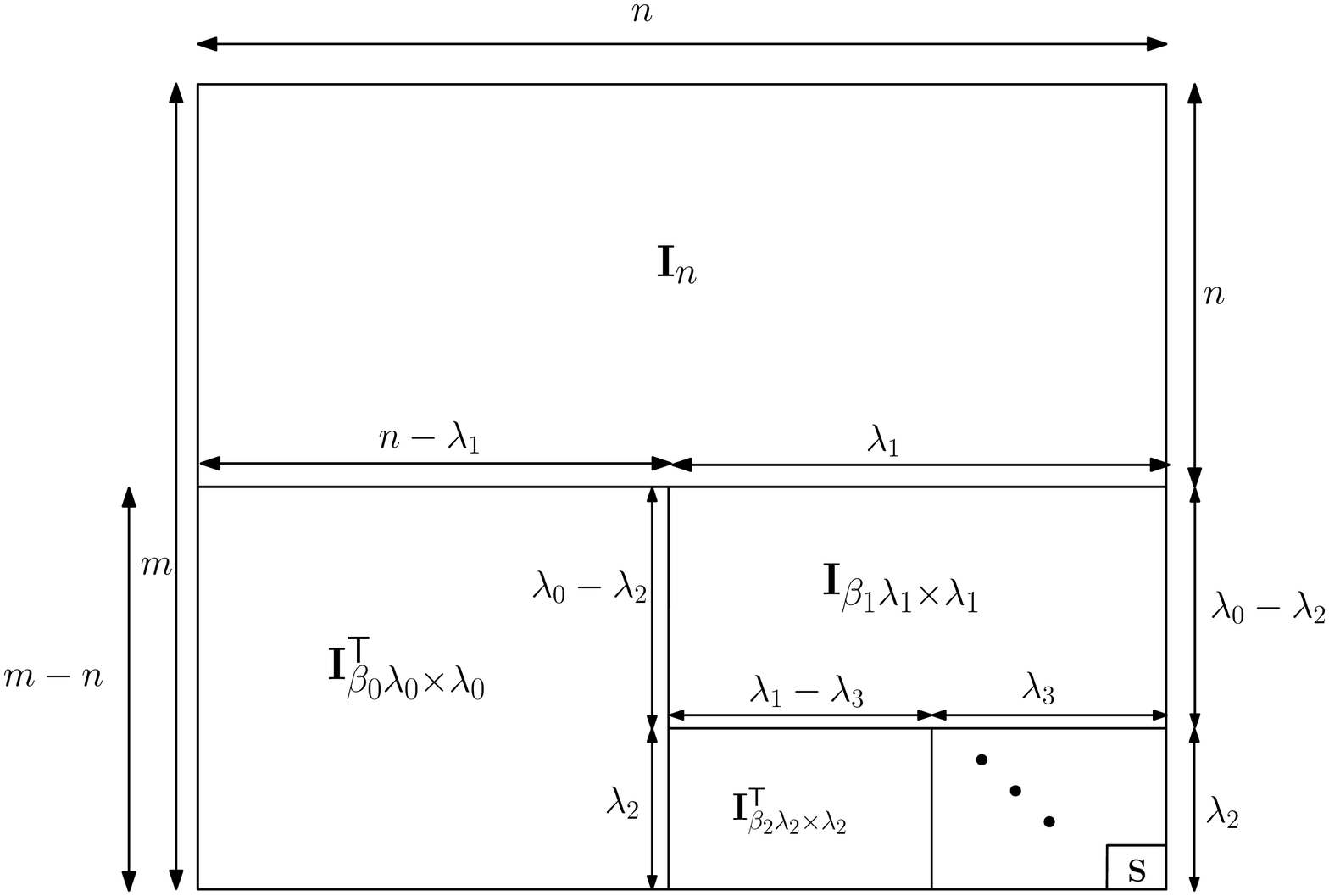}\\
~ $\mathbf{S}=\mathbf{I}_{\lambda_{l} \times \beta_l \lambda_{l}}$ if $l$ is even and ~$\mathbf{S}=\mathbf{I}_{\beta_l\lambda_{l} \times \lambda_{l}}$ otherwise.
\caption{AIR matrix of size $m \times n$.}
\label{fig1}
~ \\
\hrule
\end{figure*}
The description of the submatrices are as follows: Let $c$ and $d$ be two positive integers and $d$ divides $c$. The following matrix  denoted by $\mathbf{I}_{c \times d}$ is a rectangular  matrix.
\begin{align}
\label{rcmatrix}
\mathbf{I}_{c \times d}=\left.\left[\begin{array}{*{20}c}
   \mathbf{I}_d  \\
   \mathbf{I}_d  \\
   \vdots  \\
   \mathbf{I}_d 
   \end{array}\right]\right\rbrace \frac{c}{d}~\text{number~of}~ \mathbf{I}_d~\text{matrices}
\end{align}
and $\mathbf{I}_{d \times c}$ is the transpose of $\mathbf{I}_{c \times d}.$

Towards explaining the other quantities in the AIR matrix shown in Fig. \ref{fig1}, for a given $K,D$  and $U,$ let  $\lambda_{-1}=n,\lambda_0=m-n$ and\begin{align}
\nonumber
n&=\beta_0 \lambda_0+\lambda_1, \nonumber \\
\lambda_0&=\beta_1\lambda_1+\lambda_2, \nonumber \\
\lambda_1&=\beta_2\lambda_2+\lambda_3, \nonumber \\
\lambda_2&=\beta_3\lambda_3+\lambda_4, \nonumber \\
&~~~~~~\vdots \nonumber \\
\lambda_i&=\beta_{i+1}\lambda_{i+1}+\lambda_{i+2}, \nonumber \\ 
&~~~~~~\vdots \nonumber \\ 
\lambda_{l-1}&=\beta_l\lambda_l.
\label{chain}
\end{align}
where $\lambda_{l+1}=0$ for some integer $l,$ $\lambda_i,\beta_i$ are positive integers and $\lambda_i < \lambda_{i-1}$ for $i=1,2,\ldots,l$.

		\begin{algorithm}
		\caption{Algorithm to construct the AIR matrix $\mathbf{L}$ of size $m \times n$}
			\begin{algorithmic}[2]
				 \item Let $\mathbf{L}=m \times n$ blank unfilled matrix.
				\item [Step 1]~~~
				\begin{itemize}
				\item[\footnotesize{1.1:}] Let $m=qn+r$ for $r < n$.
				\item[\footnotesize{1.2:}] Use $\mathbf{I}_{qn \times n}$ to fill the first $qn$ rows of the unfilled part of $\mathbf{L}$.
				\item[\footnotesize{1.3:}] If $r=0$,  Go to Step 3.
				\end{itemize}

				\item [Step 2]~~~
				\begin{itemize}
				\item[\footnotesize{2.1:}] Let $n=q^{\prime}r+r^{\prime}$ for $r^{\prime} < r$.
				\item[\footnotesize{2.2:}] Use $\mathbf{I}_{q^{\prime}r \times r}^{\mathsf{T}}$ to fill the first $q^{\prime}r$ columns of the unfilled part of $\mathbf{L}$.
			    \item[\footnotesize{2.3:}] If $r^{\prime}=0$,  go to Step 3.	
				\item[\footnotesize{2.4:}] $m\leftarrow r$ and $n\leftarrow r^{\prime}$.
				\item[\footnotesize{2.5:}] Go to Step 1.
				\end{itemize}
				\item [Step 3] Exit.
		
			\end{algorithmic}
			\label{algo2}
		\end{algorithm}



\subsection{Contributions}
The contributions of this paper are summarized below:
\begin{itemize}
\item We derive a lowerbound on the broadcast rate of SUICP(SCI) to be $1+\frac{D+m}{m+1}$.
\item For the SUICP(SCI) with arbitrary $K,D,U$ and $m$, we define a set $\mathcal{\mathbf{S}}_{K,D,U,m}$ consisting of pairs of integers $(a,b)$ and prove that the rate $1+\frac{D+m+\frac{a}{b}}{m+1}$ for every $(a,b) \in \mathcal{\mathbf{S}}_{K,D,U,m}$ is achievable by $b(m+1)$-dimensional vector linear index codes and an appropriate sized AIR matrix. The constructed codes are independent of field size.    
\item We show that the constructed vector linear index codes are atmost $\frac{K~\text{mod}~(D+2m+1)}{(m+1)\left \lfloor \frac{K}{D+2m+1} \right \rfloor}$ away from the lowerbound on the broadcast rate of SUICP(SCI).
\item By using our results in \cite{VaR6}, we give an upperbound on the broadcast rate of SUICP(SCI). 
\item We obtain the capacity of SUICP(SCI) with arbitrary $K,D$ and $U$ with $m$ satisfying the relation $\text{gcd}(K,D+2m+1) \geq U+2m+1$. 
\end{itemize}

For a subset $I=\{i_1,i_2,\ldots,i_l\} \subseteq \{1,2,\ldots,K\}$, let $x_I=\{x_{i_1},x_{i_2},\ldots,x_{i_l}\}$. Let $L_k$ be the $k$th row of an AIR matrix for $k \in [0:K-1]$ and let $L_I=\{L_{i_1},L_{i_2},\ldots,L_{i_l}\}$. For vector subspaces $\mathbf{S}_{i_1},\mathbf{S}_{i_2},\ldots,\mathbf{S}_{i_l}$, let $\mathbf{S}_I=\{\mathbf{S}_{i_1},\mathbf{S}_{i_2},\ldots,\mathbf{S}_{i_l}\}$ and $<\mathbf{S}_I>=<\mathbf{S}_{i_1},\mathbf{S}_{i_2},\ldots,\mathbf{S}_{i_l}>$ denote the vector subspace spanned by all the vectors present in $\mathbf{S}_{i_1},\mathbf{S}_{i_2},\ldots,\mathbf{S}_{i_l}$. 

The remaining part of this paper is organized as follows. In Section \ref{sec2}, we derive a lowerbound on the broadcast rate of SUICP(SCI). In Section \ref{sec3}, for SUICP(SCI), we define a set $\mathbf{S}_{K,D,U,m}$ of $2$-tuples such that for every $(a,b) \in \mathcal{\mathbf{S}}_{K,D,U,m}$, the rate $1+\frac{D+m+\frac{a}{b}}{m+1}$ is achievable by using AIR matrices with vector linear index codes over every field. We also derive an upperbound on the broadcast rate of SUICP(SCI). In Section \ref{sec5}, we derive the capacity of SUICP(SCI) for some special cases. We conclude the paper in Section \ref{sec6}.
 
\section{LowerBound on the broadcast rate of SUICP(SCI)}
\label{sec2}
In this section, we derive a lowerbound on the broadcast rate of SUICP(SCI). In an $t$-dimensional vector linear index code, $x_k \in \mathbb{F}^t_q$ for $k \in [0:K-1]$. A $t$-dimensional vector linear index code of length $N$ is represented by an encoding matrix $\mathbf{L}$ $(\in \mathbb{F}^{Kt\times N}_q)$, where the $j$th column contains the coefficients used for mixing the $t$-dimensional messages $x_0,x_1,\ldots,x_{K-1}$ to get the $j$th index code symbol. Let $L_0,L_1,\ldots, L_{Kt-1}$ be the $Kt$ rows of the encoding matrix $\mathbf{L}$. Let $\mathbf{\tilde{S}}_k$ be the $t \times N$ matrix 
$$\mathbf{\tilde{S}}_{k}=\left.\left[\begin{array}{*{20}c}
   L_{kt}  \\
   L_{kt+1}  \\
   \vdots  \\
  L_{kt+t-1}   
   \end{array}\right]\right.$$
for $k \in [0:K-1]$. A codeword of the index code is 
\begin{align*}
[c_0~c_1~\ldots~c_{N-1}]=\mathbf{xL}=\sum_{i=0}^{K-1}x_i\mathbf{\tilde{S}}_k,
\end{align*}
where $\mathbf{x}=[x_{0,1}~x_{0,2}~\ldots~x_{0,t}~x_{1,1}~x_{1,2}~\ldots~x_{1,t}~\ldots~x_{K-1,t}]$. The $k$th matrix $\mathbf{\tilde{S}}_k$ $(\in \mathbb{F}^{t \times N}_q)$ contains the coefficients used for mixing the $t$-dimensional message $x_k$ in the $N$ index code symbols. That is, $<L_{kt}, L_{kt+1},\ldots, L_{kt+t-1}>$ is the subspace assigned to the message $x_k$ in $\mathbb{F}_q^N$ for $k \in [0:K-1]$.


Maleki,  Cadambe and Jafar \cite{MCJ} found the capacity of single unicast index coding problems with symmetric and neighboring side-information (SUICP(SNC)) with $K$ messages and $K$ receivers, each receiver has a total of $\tilde{U}+\tilde{D}<K$ side-information, corresponding to the $\tilde{U}$ messages before and $\tilde{D}$ messages after its desired message. In this setting, the $k$th receiver $R_{k}$ demands the message $x_{k}$ having the side-information
\begin{equation}
\label{sncantidote}
{\cal K}_k= \{x_{k-\tilde{U}},\dots,x_{k-2},x_{k-1}\}\cup\{x_{k+1}, x_{k+2},\dots,x_{k+\tilde{D}}\}.
\end{equation}

The capacity of this ICP is:
\begin{equation}
\label{snccapacity}
C=\left\{
                \begin{array}{ll}
                  {1 ~~~~~~~~~~~~ \mbox{if} ~~ \tilde{U}+\tilde{D}=K-1}\\
                  {\frac{\tilde{U}+1}{K-\tilde{D}+\tilde{U}}} ~~~ \mbox{if} ~~\tilde{U}+\tilde{D}\leq K-2. 
                  \end{array}
              \right.
\end{equation}
where $\tilde{U},\tilde{D} \in$ $\mathbb{Z},$ $0 \leq \tilde{U} \leq \tilde{D}$.

Maleki,  Cadambe and Jafar \cite{MCJ} prove the capacity given in \eqref{snccapacity} by using dimension counting outerbound. 


Suppose $\mathcal{V}_k$ denote the subspace assigned to each message $x_k$ for $k \in [0:K-1]$. Define $\alpha_j$ as follows.
\begin{align*} 
\alpha_j=\sum_{i=0}^{K-1} \text{dim}<\mathcal{V}_k,\mathcal{V}_{k+1},\ldots,\mathcal{V}_{k+j-1}>.
\end{align*}


In SUICP(SNC), every $j~(j\leq \tilde{U}+1)$ consecutive messages satisfy the property that for every receiver $R_k$ whose wanted message is in these $j$ consecutive messages, every other message in the $j$ messages is in the side-information to $R_k$. Maleki \textit{et.al} proved the following Lemma (Lemma 4 in \cite{MCJ}) to bound $\alpha_j$ in SUICP(SNC).
\begin{lemma}
\label{bound3}
For $j=2,3,\ldots,\tilde{U}+1$,
\begin{align*}
\alpha_j \geq \frac{\tilde{U}+j}{\tilde{U}+1}\alpha_1.
\end{align*}
\end{lemma}

In SUICP(SCI), every $j~(j \leq m+1)$ consecutive messages satisfy the property that for every receiver $R_k$ whose wanted message is in these consecutive $j$ messages, every other message in these $j$ consecutive messages is in the side-information to $R_k$. Hence, for SUICP(SCI), we  prove the following lemma which is analogous to Lemma \ref{bound3} for SUICP(SNC).
\begin{lemma}
\label{cor2}
For $j=2,3,\ldots,m+1$,  
\begin{align*}
\alpha_j \geq \frac{m+j}{m+1}\alpha_1.
\end{align*}
\end{lemma}
\begin{proof}
We have 
\begin{align}
\label{bound5}
\nonumber
&\sum_{i=1}^K \text{dim}<\mathcal{V}_i,\mathcal{V}_{i+1},\ldots,\mathcal{V}_{i+j-2},\mathcal{V}_{i+j-2+m+1}>\\& =\underbrace{\sum_{i=1}^K \text{dim}<\mathcal{V}_i,\mathcal{V}_{i+1},\ldots ,\mathcal{V}_{i+j-2}>}_{\alpha_{j-1}}+\underbrace{\sum_{i=1}^K\text{dim}( \mathcal{V}_{i+j-2+m+1})}_{\alpha_1}. 
\end{align}

\eqref{bound5} follows from the fact that for the receiver $R_{i+j-2+m+1}$, the message symbols $x_i,x_{i+1},\ldots,x_{i+j-2}$ are in interference and hence $<\mathcal{V}_i,\mathcal{V}_{i+1}, \ldots \mathcal{V}_{i+j-2}> \cap \mathcal{V}_{i+j-2+m+1}=\phi$.
 
In \eqref{bound6}, we upperbound $$\sum_{i=1}^K \text{dim}<\mathcal{V}_i,\mathcal{V}_{i+1},\ldots,\mathcal{V}_{i+j-2},\mathcal{V}_{i+j-2+m+1}>$$ by using the submodular property of dimension function (for subspaces $\mathcal{V}_{i_1}$ and $\mathcal{V}_{i_2}$, we have dim($\mathcal{V}_{i_1})$+dim($\mathcal{V}_{i_2})$=dim($\mathcal{V}_{i_1} \cup \mathcal{V}_{i_2})$+dim($\mathcal{V}_{i_1}\cap \mathcal{V}_{i_2})$).
\begin{figure*}
\begin{align}
\label{bound6}
\nonumber
&\sum_{i=1}^K \text{dim}<\mathcal{V}_i,\mathcal{V}_{i+1},\ldots,\mathcal{V}_{i+j-2},\mathcal{V}_{i+j-2+m+1}> ~~~\leq \underbrace{\sum_{i=1}^K \text{dim}<\mathcal{V}_i,\mathcal{V}_{i+1},\ldots, \mathcal{V}_{i+j-2},\mathcal{V}_{i+j-1}>}_{\alpha_j}\\& \nonumber~~~~~~~~~~~~~~~~~~~~~~~~+
\sum_{i=1}^K \text{dim}<\mathcal{V}_{i+1},\mathcal{V}_{i+2},\ldots,\mathcal{V}_{i+j-1},\mathcal{V}_{i+j+m-1}>-
\underbrace{\sum_{i=1}^K \text{dim}<\mathcal{V}_{i+1},\mathcal{V}_{i+2},\ldots,\mathcal{V}_{i+j-1}>}_{\alpha_{j-1}} \\& \nonumber \leq \underbrace{\sum_{i=1}^K \text{dim}<\mathcal{V}_i,\mathcal{V}_{i+1},\ldots,\mathcal{V}_{i+j-2},\mathcal{V}_{i+j-1}>}_{\alpha_j}+ \underbrace{\sum_{i=1}^K \text{dim}<\mathcal{V}_{i+1},\mathcal{V}_{i+2},\ldots,\mathcal{V}_{i+j}>}_{\alpha_j}\\& \nonumber 
+\sum_{i=1}^K \text{dim}<\mathcal{V}_{i+2},\mathcal{V}_{i+3},\ldots,\mathcal{V}_{i+j},\mathcal{V}_{i+j+m-1}>-
\underbrace{\sum_{i=1}^K \text{dim}<\mathcal{V}_{i+2},\mathcal{V}_{i+3},\ldots,\mathcal{V}_{i+j}>}_{\alpha_{j-1}}-
\underbrace{\sum_{i=1}^K \text{dim}<\mathcal{V}_{i+1},\ldots,\mathcal{V}_{i+j-1}>}_{\alpha_{j-1}} \leq \\& \nonumber ~~~~~~~~~~~~~~~~\vdots~~~~~~~~~~~~~~~~~~~~~~~\vdots \\& \nonumber  \leq \underbrace{\underbrace{\sum_{i=1}^K \text{dim}<\mathcal{V}_i, \ldots,\mathcal{V}_{i+j-1}>}_{\alpha_j}+ \underbrace{\sum_{i=1}^K \text{dim}<\mathcal{V}_{i+1},\mathcal{V}_{i+2},\ldots,\mathcal{V}_{i+j}>}_{\alpha_j}+\ldots+\underbrace{\sum_{i=1}^K \text{dim}<\mathcal{V}_{i+m},\mathcal{V}_{i+m+1},\ldots,\mathcal{V}_{i+j+m-1}>}_{\alpha_j}}_{(m+1)~\text{terms}} \\& -
\underbrace{\underbrace{\sum_{i=1}^K \text{dim}<\mathcal{V}_{i+m},\mathcal{V}_{i+m+1},\ldots,\mathcal{V}_{i+j+m-2}>}_{\alpha_{j-1}}-\ldots -
\underbrace{\sum_{i=1}^K \text{dim}<\mathcal{V}_{i+2},\mathcal{V}_{i+3},\ldots,\mathcal{V}_{i+j}>}_{\alpha_{j-1}}-\underbrace{\sum_{i=1}^K \text{dim}<\mathcal{V}_{i+1},\ldots,\mathcal{V}_{i+j-1}>}_{\alpha_{j-1}}}_{m~\text{terms}}.
\end{align}
\end{figure*}
From \eqref{bound5} and \eqref{bound6}, we have 
\begin{align*}
\alpha_{j-1}+\alpha_1 \leq (m+1)\alpha_j-m\alpha_{j-1}.
\end{align*}
Hence, we have
\begin{align*}
(m+1)\alpha_j &\geq (m+1)\alpha_{j-1}+\alpha_1 \\& \geq (m+1)\alpha_{j-2}+2\alpha_1 \\& \geq (m+1)\alpha_{j-3}+3\alpha_1 \\& \geq \ldots \geq (m+1)\alpha_1+(j-1)\alpha_1.
\end{align*}
This proves Lemma \ref{cor2}.
\end{proof}
\begin{lemma}
\label{cor3}
For SUICP(SCI) with arbitrary $K,D,U$ and $m$, define $t$ and $j$ as follows,
\begin{align*}
&\text{if}~D~\text{mod}~(m+1) = 0 \\&
~~~~~~~~~~~~~t= \frac{D}{m+1}-1,~~j=m+1, \\&
\text{if}~D~\text{mod}~(m+1) \neq 0\\&
~~~~~~~~~~~~~t=\left \lfloor \frac{D}{m+1}\right \rfloor,~~j=D~\text{mod}~(m+1). \\&
\end{align*}
then, $\alpha_{D} \geq t\alpha_1+\alpha_j$.
\end{lemma}
\begin{proof}
\begin{align*}
\alpha_{D}&=\alpha_{tm+j}\\&=\sum_{i=1}^K \text{dim}<\mathcal{V}_i,\mathcal{V}_{i+1},\ldots,\mathcal{V}_{tm+j-1}>\\& \geq \sum_{i=1}^K \text{dim}<\mathcal{V}_i,\mathcal{V}_{i+1},\ldots,\mathcal{V}_{i+j-1},\mathcal{V}_{i+j-1+m+1}, \\& ~~~~~~~~~ \mathcal{V}_{i+j-1+2(m+1)},\ldots,\mathcal{V}_{i+j-1+t(m+1)}>\\&=\sum_{i=1}^K \Big\{\text{dim}<\mathcal{V}_i,\mathcal{V}_{i+1},\ldots,\mathcal{V}_{i+j-1}>\\&~~~~~~~~~~+\text{dim}(\mathcal{V}_{i+j-1+m+1})+    \text{dim}(\mathcal{V}_{i+j-1+2(m+1)})\\&~~~~~~~~~~+\ldots+\text{dim}(\mathcal{V}_{i+j-1+t(m+1)})\Big\}\\&= t\alpha_1+\alpha_j.
\end{align*}
\end{proof}
\begin{theorem}
\label{lowerbound}
Consider a SUICP(SCI) with arbitrary $K,D,U$ and $m$. The broadcast rate $\beta$ of this ICP is lowerbounded by \begin{align*}
\beta \geq 1 + \frac{D+m}{m+1}.
\end{align*}
\end{theorem}
\begin{proof}
From Lemma \ref{cor2} and \ref{cor3}, we have 
\begin{align*}
\alpha_{D}  \geq t\alpha_1+\alpha_j & \geq  t\alpha_1 + \frac{m+j}{m+1}\alpha_1 \\&= \frac{t(m+1)+j+m}{m+1}\alpha_1=\frac{D+m}{m+1}\alpha_1.
\end{align*}

Then the dimension of the desired message plus the dimension of the interference at all receivers should be less than or equal to 
\begin{align*}
\alpha_{D}+\alpha_1  \geq \frac{D+m}{m+1}\alpha_1 + \alpha_1.
\end{align*}

The broadcast rate should be greater than or equal to the number of dimensions required by wanted message and interference per message symbol. Hence, we have
\begin{align*}
\beta \geq \frac{\alpha_{D}+\alpha_1}{\alpha_1}  \geq 1+ \frac{D+m}{m+1}.
\end{align*}
This completes the proof.
\end{proof}

\section{Near-Optimal Vector Linear Index Codes of SUICP(SCI)}
\label{sec3}
In this section, we define a set $\mathbf{S}_{K,D,U,m}$ of $2$-tuples such that for every $(a,b) \in \mathcal{\mathbf{S}}_{K,D,U,m}$, the rate $1+\frac{D+m+\frac{a}{b}}{m+1}$ is achievable by using AIR matrices with vector linear index codes over every field. Constructing vector linear index codes for SUICP(SCI) follows from the two steps given below.
\begin{itemize}
\item First we convert the SUICP(SCI) problem into SUICP(SNI) problem by converting the $m+1$ message symbols into one extended message symbol.
\item We construct vector linear index codes for SUICP(SCI) by constructing vector linear index codes for SUICP(SNI).
\end{itemize}

In \cite{VaR5}, we proved the following interference alignment Lemma for SUICP(SNI).
\begin{lemma}
\label{lemma1}
Consider a SUICP(SNI) with $K$ messages, $D$ and $U$ interfering messages after and before the desired message. Let $\mathbf{L}$ be a $t$-dimensional vector encoding matrix (not necessarily an AIR matrix) of size $Kt \times N$ for this index coding problem. Let the vector index code be generated by multiplying $Kt$ message symbols $[x_{0,1}~x_{0,2}~\ldots~x_{0,t}~x_{1,1}~x_{1,2}~\ldots~x_{1,t}~\ldots~x_{K-1,t}]$ with the encoding matrix $\mathbf{L}$. Let $\mathbf{S}_k$ be the row space of $\tilde{\mathbf{S}}_k$ and $\mathbf{S}_{k \setminus i}$ be the subspace spanned by the $t-1$ rows in $\tilde{\mathbf{S}}_k$ after deleting the row $L_{kt+i}$. The receiver $R_k$ can decode $x_{k,1},x_{k,2},\ldots,x_{k,t}$ if and only if 
\begin{align} 
\label{sind}
L_{kt+i} \notin  < \mathbf{S}_{{\cal{I}}_k},\mathbf{S}_{k \setminus i} >
\end{align}
for $k\in [0:K-1]$, $i \in [0:t-1]$.
\end{lemma}
\begin{corollary}
\label{cor1}
Consider the SUICP(SNI) with $K$ messages, $D$ interference after and $U$ interference before the desired message. Let $\mathbf{L}$ be a scalar linear encoding matrix (not necessarily an AIR matrix) of size $K \times N$ for this index coding problem. The receiver $R_k$ can decode $x_k$ if and only if 
\begin{align} 
\label{sind2}
L_{k} \notin  < L_{{\cal{I}}_k}>
\end{align}
for $k\in [0:K-1]$.
\end{corollary}
In \cite{VaR5}, we proved the following property of an AIR matrix. We use Lemma \ref{lemma3} in the proof of Theorem \ref{thm1}.
\begin{lemma}
\label{lemma3}
In the AIR matrix of size $g \times h$ for some positive integers $g$ and $h \leq g$, every row $L_k$ is not in the span of $h-1$ rows above and $\text{gcd}(g,h)-1$ rows below $L_k$ for $k \in [0:g-1]$.
\end{lemma}

\begin{definition}
\label{def1}
Consider a SUICP(SCI) with arbitrary $K,D,U$ and $m$. Define the set $\mathbf{S}_{K,D,U,m}$ as  
\begin{align}
\label{ab}
\nonumber
\mathbf{S}_{K,D,U,m}=\{(a,b):\text{gcd}(bK,b(D+&2m+1)+a)\geq \\& b(U+2m+1)\}
\end{align}
for $a \in Z_{\geq 0}$ and $b \in Z_{>0}$. 
\end{definition}
In Theorem \ref{thm1}, we convert the SUICP(SCI) into SUICP(SNI) and give the construction of near-optimal vector linear index codes for SUICP(SCI) with arbitrary $K,U,D$ and $m$.

\begin{theorem}
\label{thm1}
Given arbitrary positive integers $K,D,U$ and $m$, consider the SUICP(SCI) with $K$ messages $\{x_0,x_1,\cdots,x_{K-1}\}$ and the receivers being $K,$ and the receiver $R_k$ for $k \in [0:K-1]$ wanting the message $x_k$ and having the interference given by 
\begin{align}
\label{thm1eq}
\nonumber
{\cal I}_k= &\{x_{k-U-m},\dots,x_{k-m-2},x_{k-m-1}\}~\cup \\& \{x_{k+m+1},x_{k+m+2},\dots,x_{k+m+D}\}.
\end{align}
Let $(a,b) \in \mathbf{S}_{K,D,U,m}$ and $x_k=(x_{k,1},x_{k,2},\cdots,x_{k,b(m+1)})$ be the message vector wanted by the $k$th receiver $R_k$, where $x_k \in \mathbb{F}_q^{b(m+1)}, x_{k,i} \in \mathbb{F}_q$ for every $k \in [0:K-1]$ and $i \in [1:b(m+1)]$. Let
\begin{align}
\label{thm1eq6}
y_{k,j}=\sum_{i=1}^{m+1} x_{k+1-i,(j-1)(m+1)+i} 
\end{align}
for $k \in [0:K-1]$ and $j \in [1:b]$. Let $\mathbf{L}$ be the AIR matrix of size $Kb \times (b(D+2m+1)+a)$ and $L_k$ be the $k$th row of $\mathbf{L}$ for every $k \in [0:Kb-1]$. The $b(m+1)$ dimensional vector linear index code for the given SUICP(SCI) is given by 
\begin{align}
\label{thm1eq5}
[c_0~c_1~\ldots~c_{b(D+2m+1)+a-1}]=\sum_{k=0}^{K-1}\sum_{j=1}^{b}y_{k,j}L_{kb+j-1}.
\end{align} 
\end{theorem}

\begin{proof}
We prove that every receiver $R_k$ for $k \in [0:K-1]$ decodes its $b(m+1)$ wanted message symbols $x_{k,1},x_{k,2},\ldots,x_{k,b(m+1)}$. We have
\begin{align*}
y_{k,j}=\sum_{i=1}^{m+1} x_{k+1-i,(j-1)(m+1)+i} 
\end{align*}
for $k \in [0:K-1]$ and $j \in [1:b]$. We refer $y_{k,j}$ as extended message symbol. The alignment of extended message symbols $y_{k,1}$ for $k \in [0:K-1]$ is shown in Fig. \ref{vcfig1}.
\begin{figure*}[t]
\centering
\includegraphics[scale=0.45]{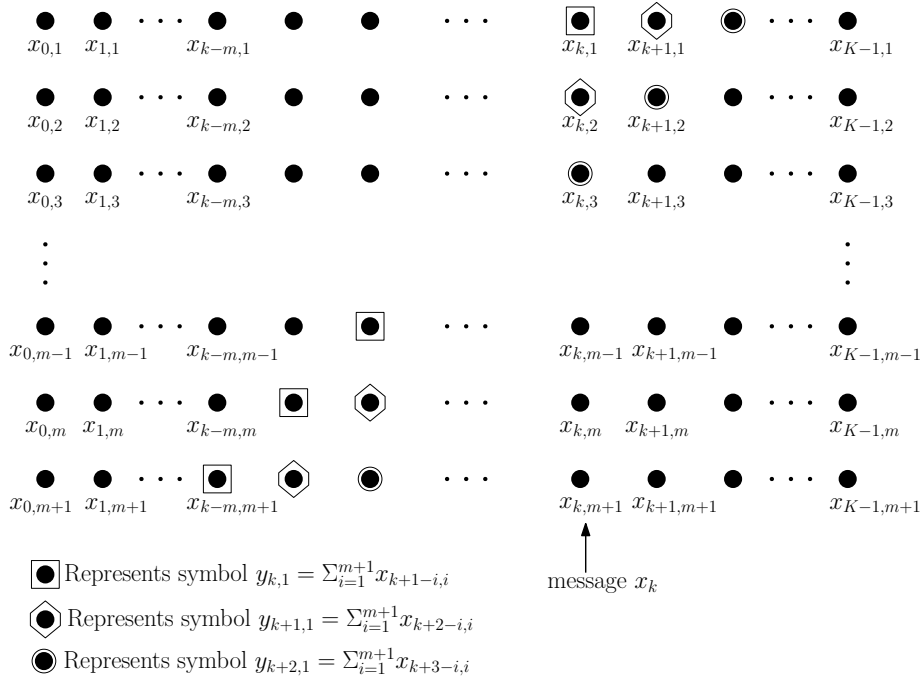}
\caption{Alignment of the symbols $y_{k,j=1}$ for $s \in [0:K-1]$.}
\label{vcfig1}
\hrule
\end{figure*}

The extended message symbol $y_{k,j}$ comprises of m+1 message symbols $x_{k,(j-1)(m+1)+1}$, $x_{k-1,(j-1)(m+1)+2}$,$\ldots$, $x_{k-m,(j-1)(m+1)+m+1}$. The $Kb(m+1)$ message symbols $x_{k,i}$ for $k \in [0:K-1]$ and $i \in [1:b(m+1)]$ appear exactly once in $Kb$ extended message symbols $y_{k,j}$ for $k \in [0:K-1]$ and $j \in [1:b]$.

We say that the extended message symbol $y_t$ is in the interference of $R_k$ if 
\begin{itemize}
\item $y_t$ comprises of atleast one message symbol belonging to $\mathcal{I}_k$ (or)
\item $y_t$ comprises of a message symbol belonging to $x_k$ and $R_k$ is yet to decode this message symbol. 
\end{itemize}

For the receiver $R_k$, its wanted message symbol $x_{k,(j-1)(m+1)+t}$ is present in in the extended message symbol $y_{k+t-1,j}$ for every $j \in [1:b]$ and $t \in [1:m+1]$. We have 
\begin{align*}
y_{k+t-1,j}=\sum_{i=1}^{m+1} x_{k+t-i,(j-1)(m+1)+i}.
\end{align*}

In $y_{k+t-1,j}$, every message symbol other than  $x_{k,(j-1)(m+1)+t}$ is in side-information of $R_k$ for every $k \in [0:K-1],j \in [1:b]$ and $t \in [1:m+1]$.

From the the alignment of the extended message symbols $y_{k,j}$ for every $k \in [0:K-1]$ and $j \in [1:b]$, we can conclude the following observations regarding the number of interfering extended messages after $y_{k-t+1,j}$.  

\begin{itemize}
\item In $y_{k+s,j}$, every message symbol is in the side-information of $R_k$ except $x_{k,(j-1)(m+1)+s+1}$ for $s \in [t:m]$. If $R_k$ has not decoded $x_{k,(j-1)(m+1)+s+1}$, then $y_{k+s,j}$ acts as an interference to $R_k$. There exists atmost $b(m-t)$ interfering extended message symbols like this.
\item In $y_{k+m+s,j}$, the message symbol $x_{k+m+s,(j-1)(m+1)+1}$ is in the interference 
to $R_k$ for every $j \in [1:b]$ and $s \in [1:D]$. There exists $bD$ interfering extended message symbols like this.
\item In $y_{k+m+D+s,j}$, the message symbol $x_{k+m+D,(j-1)(m+1)+s+1}$ is in the interference to $R_k$ for every $j \in [1:b]$ and $s \in [1:m]$. There exists $bm$ interfering extended message symbols like this.
\item In $y_{k+t-1,j^\prime}$ for $j^\prime \in [1:b]$ and $j^\prime \neq j$, the message symbol $x_{k,(j^\prime-1)(m+1)+t}$ is present and it contributes to interference to $R_k$ if $R_k$ has not decoded it before. There exists atmost $b-1$ interfering extended message symbols like this.
\end{itemize}

From the the alignment of the extended message symbols $y_{k,j}$ for every $k \in [0:K-1]$ and $j \in [1:b]$, we can conclude the following observations regarding the number of interfering extended messages before $y_{k-t+1,j}$. 
\begin{itemize}
\item In $y_{k+s,j}$ for $s \in [0:t-2]$, every message symbol is in side-information of $R_k$ except $x_{k,(j-1)(m+1)+s+1}$. If $R_k$ has not decoded $x_{k,(j-1)(m+1)+s+1}$, then $y_{k+s,j}$ acts as an interference to $R_k$. There exists atmost $bt$ interfering extended message symbols like this.
\item In the extended message symbol $y_{k-s,j}$, the message symbol $x_{k-m-s,(j-1)(m+1)+m+1}$ is in the interference 
to $R_k$ for every $j \in [1:b]$ and $s \in [1:U]$. There exists $bU$ interfering extended message symbols like this.
\item In the extended message symbol $y_{k-U-s,j}$, the message symbol $x_{k-m-U,(j-1)(m+1)+m+1-s}$ is in the interference to $R_k$ for every $j \in [1:b]$ and $s \in [1:m]$. There exists $bm$ interfering extended message symbols like this.
\item In $y_{k+t-1,j^\prime}$ for $j^\prime \in [1:b]$ and $j^\prime \neq j$, the message symbol $x_{k,(j^\prime-1)(m+1)+t}$ is present and it contributes to interference to $R_k$ if $R_k$ has not decoded it before. There exists atmost $b-1$ interfering extended message symbols like this.
\end{itemize}

For $t \in [0:K-1]/[k-m-U:k+D+2m]$, every message symbol in $y_{t,j}$ is in side-information of $R_k$ for every $k \in [0:K-1]$ and $j \in [1:b]$.

Hence, there exists atmost $$b(m-t)+bD+bm+b-1$$ extended message symbols after $y_{k+t-1,j}$ and $$bt+bU+bm+b-1$$ extended message below  $y_{k+t-1,j}$ such that every extended message symbol comprises of atleast one message symbol which belongs to $x_k \cup \mathcal{I}_k$. 

If $t=m+1$, the receiver $R_k$ sees $b(D+m)+b-1$ interfering messages after and $b(U+2m)+b-1$ interfering extended message before the desired extended message. If $t=1$, the receiver $R_k$ sees $b(D+2m)+b-1$ interfering messages after and $b(U+m)+b-1$ interfering extended messages before the desired extended message. For any $t \in [1:m+1]$, the receiver $R_k$ sees atmost $b(D+2m+1)-1$ interfering messages after and $b(U+2m+1)-1$ interfering extended messages before the desired extended message.

Given $(a,b) \in \mathbf{S}_{K,D,U,m}$ and $\mathbf{L}$ is the AIR matrix of size $Kb \times (b(D+2m+1)+a)$. From Lemma \ref{lemma3}, every row $L_k$ is not in the span of $b(D+2m+1)+a$ after and $\text{gcd}(Kb,b(D+2m+1)+a)-1$ rows before to $L_k$ for every $k \in [0:Kb-1]$. From Definition \ref{def1}, $a$ and $b$ are the positive integers satisfying the relation
\begin{align*}
\text{gcd}(Kb,b(D+2m+1)+a)=b(U+2m+1)+c
\end{align*}
for some $c \in Z_{\geq 0}$. Hence, every row $L_k$ in $\mathbf{L}$ is not in the span of $b(D+2m+1)+a-1$ rows above and $b(U+2m+1)-1$ rows below to $L_k$ for each $k \in [0:Kb-1]$. According to Lemma \ref{lemma1},  the matrix $\mathbf{L}$ can be used as a $b$-dimensional encoding matrix for SUICP(SNI) with $K$ messages, $b(D+2m)+b-1$ and $b(U+2m)+b-1$ interfering messages after and before. Hence, the receiver $R_k$ decodes $y_{k+t-1,j}$ and then decodes $x_{k,(j-1)(m+1)+t}$ for every $t \in [1:m+1],j \in [1:b]$.

For every message symbol $x_{k,(j-1)(m+1)+i}$ for $i \in [1:m+1]$, the decoding is performed as given below.
\begin{itemize}
\item $R_k$ first decodes the extended message symbol $y_{k-i+1,j}$ for $i \in [1:m+1]$, where the message symbol $x_{k,(j-1)(m+1)+i}$ is present.
\item In $y_{k-i+1,j}$, every message symbol present is in the side-information of $R_k$. Hence, $R_k$ decodes its wanted message symbol $x_{k,(j-1)(m+1)+i}$ from $y_{k-i+1,j}$.
\end{itemize}

The matrix $\mathbf{L}$ is mapping $Kb(m+1)$ message symbols into $b(D+2m+1)+a$ broadcast symbols. The rate achieved by the proposed encoding scheme is given by
\begin{align*}
\frac{b(D+2m+1)+a}{b(m+1)}=1+\frac{D+m+\frac{a}{b}}{m+1}.
\end{align*}
\end{proof}

In Lemma \ref{lemma6}, we show that the rate achieved by vector linear index codes constructed in Theorem \ref{thm1} are atmost  $\frac{K \text{mod} (D+2m+1)}{(m+1)\left\lfloor\frac{K}{D+2m+1}\right\rfloor}$ away from the lowerbound on broadcast rate of SUICP(SCI).
\begin{lemma}
\label{lemma6}
For every  SUICP(SCI) with arbitrary $K,D,U$ and $m$, there exists $(a,b) \in \mathcal{\mathbf{S}}_{K,D,U,m}$ such that
\begin{align}
\label{gcd84}
1+\frac{D+m+\frac{a}{b}}{m+1} \leq \frac{K}{(m+1)\left \lfloor \frac{K}{D+2m+1} \right \rfloor}.
\end{align}
\end{lemma}
\begin{proof}
We have 
\begin{align*}
K=\left\lfloor\frac{K}{D+2m+1}\right\rfloor(D+2m+1)+K \text{mod} (D+2m+1).
\end{align*}
Hence, 
\begin{align}
\label{gcd83}
\frac{K}{\left\lfloor\frac{K}{D+2m+1}\right\rfloor}=D+2m+1+\frac{K \text{mod} (D+2m+1)}{\left\lfloor\frac{K}{D+2m+1}\right\rfloor}.
\end{align}
From \eqref{gcd83}, we have 
\begin{align}
\label{gcd841}
\nonumber
\frac{K}{(m+1)\left\lfloor\frac{K}{D+2m+1}\right\rfloor}&=\underbrace{1+\frac{D+m}{m+1}}_{\text{lowerbound on}~\beta}+\frac{K \text{mod} (D+2m+1)}{(m+1)\left\lfloor\frac{K}{D+2m+1}\right\rfloor}\\&=1+\frac{D+m}{m+1}+\frac{\alpha}{(m+1)\gamma},
\end{align}
where $\alpha=K \text{mod} (D+2m+1)$ and $\gamma=\left\lfloor\frac{K}{D+2m+1}\right\rfloor$.

From \eqref{gcd83}, we have
\begin{align}
\label{gcd81}
\alpha=K-\gamma(D+2m+1)
\end{align}
and these values of $\alpha$ and $\gamma$ satisfy the equation $\text{gcd}(K \gamma,\gamma(D+2m+1)+\alpha)\geq \gamma (U+2m+1)$. Hence, $(\alpha,\gamma) \in \mathcal{\mathbf{S}}_{K,D,U,m}$. This completes the proof.
\end{proof}
\begin{corollary}
\label{cor4}
In SUICP(SCI), the vector linear index codes constructed by AIR matrices are within $\frac{K \text{mod} (D+2m+1)}{(m+1)\left\lfloor\frac{K}{D+2m+1}\right\rfloor}$ symbols per message from the lowerbound on broadcast rate given in Theorem \ref{lowerbound}.
\end{corollary}

\begin{example}
\label{ex1}
Consider a SUICP(SCI) with $K=18,D=7,U=1$ and $m=2$. For this SUICP(SCI), we have $D+2m=11,U+2m=5$ and $\text{gcd}(K,D+2m+1)=\text{gcd}(18,12)=6 \geq U+2m+1$. Hence, for this SUICP(SCI), $(a,b)=(0,1) \in \mathbf{S}_{K,D,U,m}$. The rate achieved by proposed construction is
\begin{align*}
1+\frac{D+m+\frac{a}{b}}{m+1}=1+\frac{7}{3}+\frac{2}{3}=4.
\end{align*}
We have $y_k=x_{k,1}+x_{k-1,2}+x_{k-2,3}$ for $k \in [0:17]$. The index code for this SUICP(SCI) is obtained by 
\begin{align*}
[c_0~c_1~ \ldots ~c_{11}]=[y_0~y_1~\ldots~y_{17}]\mathbf{L}_{18 \times 12},
\end{align*}
where $\mathbf{L}_{18 \times 12}$ is the AIR matrix of size $18 \times 12$ as given below.

\arraycolsep=0.5pt
\setlength\extrarowheight{-3.0pt}
{
$$\mathbf{L}_{18 \times 12}=\left[
\begin{array}{cccccccccccccc}
1 & 0 & 0 & 0 & 0 & 0 & 0 & 0 & 0 & 0 & 0 & 0\\
0 & 1 & 0 & 0 & 0 & 0 & 0 & 0 & 0 & 0 & 0 & 0\\
0 & 0 & 1 & 0 & 0 & 0 & 0 & 0 & 0 & 0 & 0 & 0\\
0 & 0 & 0 & 1 & 0 & 0 & 0 & 0 & 0 & 0 & 0 & 0\\
0 & 0 & 0 & 0 & 1 & 0 & 0 & 0 & 0 & 0 & 0 & 0\\
0 & 0 & 0 & 0 & 0 & 1 & 0 & 0 & 0 & 0 & 0 & 0\\
0 & 0 & 0 & 0 & 0 & 0 & 1 & 0 & 0 & 0 & 0 & 0\\
0 & 0 & 0 & 0 & 0 & 0 & 0 & 1 & 0 & 0 & 0 & 0\\
0 & 0 & 0 & 0 & 0 & 0 & 0 & 0 & 1 & 0 & 0 & 0\\
0 & 0 & 0 & 0 & 0 & 0 & 0 & 0 & 0 & 1 & 0 & 0\\
0 & 0 & 0 & 0 & 0 & 0 & 0 & 0 & 0 & 0 & 1 & 0\\
0 & 0 & 0 & 0 & 0 & 0 & 0 & 0 & 0 & 0 & 0 & 1\\
1 & 0 & 0 & 0 & 0 & 0 & 1 & 0 & 0 & 0 & 0 & 0\\
0 & 1 & 0 & 0 & 0 & 0 & 0 & 1 & 0 & 0 & 0 & 0\\
0 & 0 & 1 & 0 & 0 & 0 & 0 & 0 & 1 & 0 & 0 & 0\\
0 & 0 & 0 & 1 & 0 & 0 & 0 & 0 & 0 & 1 & 0 & 0\\
0 & 0 & 0 & 0 & 1 & 0 & 0 & 0 & 0 & 0 & 1 & 0\\
0 & 0 & 0 & 0 & 0 & 1 & 0 & 0 & 0 & 0 & 0 & 1\\
 \end{array}
\right]$$
}

The 12 broadcast symbols for this SUICP(SCI) are given in Table \ref{table1} below. 

\begin{table}[ht]
\centering
\setlength\extrarowheight{0pt}
\begin{tabular}{|c|}
\hline
$c_0=\underbrace{x_{0,1}+x_{17,2}+x_{16,3}}_{y_0}+\underbrace{x_{12,1}+x_{11,2}+x_{10,3}}_{y_{12}}$\\
\hline
$c_1=\underbrace{x_{1,1}+x_{0,2}+x_{17,3}}_{y_1}+\underbrace{x_{13,1}+x_{12,2}+x_{11,3}}_{y_{13}}$ \\
\hline
$c_2=\underbrace{x_{2,1}+x_{1,2}+x_{0,3}}_{y_2}+\underbrace{x_{14,1}+x_{13,2}+x_{12,3}}_{y_{14}}$ \\
\hline
$c_3=\underbrace{x_{3,1}+x_{2,2}+x_{1,3}}_{y_3}+\underbrace{x_{15,1}+x_{14,2}+x_{13,3}}_{y_{15}}$ \\
\hline
$c_4=\underbrace{x_{4,1}+x_{3,2}+x_{2,3}}_{y_4}+\underbrace{x_{16,1}+x_{15,2}+x_{15,3}}_{y_{16}}$ \\
\hline
$c_5=\underbrace{x_{5,1}+x_{4,2}+x_{3,3}}_{y_5}+\underbrace{x_{17,1}+x_{16,2}+x_{15,3}}_{y_{17}}$ \\
\hline
$c_6=\underbrace{x_{6,1}+x_{5,2}+x_{4,3}}_{y_6}+\underbrace{x_{12,1}+x_{11,2}+x_{10,3}}_{y_{12}}$\\
\hline
$c_7=\underbrace{x_{7,1}+x_{6,2}+x_{5,3}}_{y_7}+\underbrace{x_{13,1}+x_{12,2}+x_{11,3}}_{y_{13}}$ \\
\hline
$c_8=\underbrace{x_{8,1}+x_{7,2}+x_{6,3}}_{y_8}+\underbrace{x_{14,1}+x_{13,2}+x_{12,3}}_{y_{14}}$ \\
\hline
$c_9=\underbrace{x_{9,1}+x_{8,2}+x_{7,3}}_{y_9}+\underbrace{x_{15,1}+x_{14,2}+x_{13,3}}_{y_{15}}$ \\
\hline
$ c_{10}=\underbrace{x_{10,1}+x_{9,2}+x_{8,3}}_{y_{10}}+\underbrace{x_{16,1}+x_{15,2}+x_{14,3}}_{y_{16}}$  \\
\hline
$c_{11}=\underbrace{x_{11,1}+x_{10,2}+x_{9,3}}_{y_{11}}+\underbrace{x_{17,1}+x_{16,2}+x_{15,3}}_{y_{17}}$ \\
\hline
\end{tabular}
\vspace{5pt}
\caption{Vector linear index code for SUICP(SNI) given in Example \ref{ex1}}
\label{table1}
\vspace{-5pt}
\end{table}

Receiver $R_k$ required to decode three message symbols $x_{k,1},x_{k,2}$ and $x_{k,3}$ for $k \in [0:17]$. Let $\tau_{k,j}$ be the code symbols used by receiver $R_k$ to decode $x_{k,j}$ for $k \in [0:17]$ and $j \in [1:3]$. Table \ref{table2} gives the code symbols used by each receiver to decode its wanted message symbol.
\begin{table}[ht]
\centering
\setlength\extrarowheight{0pt}
\begin{tabular}{|c|c|c|c|c|c|}
\hline
$x_{k,1}$ &$\tau_{k,1}$&$x_{k,2}$&$\tau_{k,2}$&$x_{k,3}$&$\tau_{k,3}$ \\
\hline
$x_{0,1}$ & $c_0$ & $x_{0,2}$&$c_1$&$x_{0,3}$&$c_2$ \\
\hline
$x_{1,1}$ & $c_1$ & $x_{1,2}$&$c_2$&$x_{1,3}$&$c_3$  \\
\hline
$x_{2,1}$ & $c_2$ & $x_{2,2}$&$c_3$&$x_{2,3}$ &$c_4$\\
\hline
$x_{3,1}$ & $c_3$ &$x_{3,2}$&$c_4$&$x_{3,3}$ &$c_5$\\
\hline
$x_{4,1}$ & $c_4$ &$x_{4,2}$&$c_5$&$x_{4,3}$&$c_0+c_6$\\
\hline
$x_{5,1}$ & $c_5$ & $x_{5,2}$&$c_0+c_6$ &$x_{5,3}$&$c_1+c_7$\\
\hline
$x_{6,1}$ & $c_0+c_6$ & $x_{6,2}$& $c_1+c_7$ &$x_{6,3}$& $c_2+c_8$\\
\hline
$x_{7,1}$ & $c_1+c_7$ & $x_{7,2}$& $c_2+c_8 $&$x_{7,3}$& $c_3+c_9$ \\
\hline
$x_{8,1}$& $c_2+c_8$ & $x_{8,2}$& $c_3+c_9$ &$x_{8,3}$& $c_4+c_{10} $ \\
\hline
$x_{9,1}$ & $c_3+c_9$ & $x_{9,2}$& $c_4+c_{10} $ &$x_{9,3}$& $c_5+c_{11}$  \\
\hline
$x_{10,1}$ & $c_4+c_{10} $ & $x_{10,2}$& $c_5+c_{11} $ &$x_{10,3}$& $c_6$  \\
\hline
$x_{11,1}$ & $c_5+c_{11} $ & $x_{11,2}$& $c_6$ &$x_{11,3}$&$c_7$\\
\hline
$x_{12,1}$ & $c_6$ & $x_{12,2}$&$c_7$&$x_{12,3}$&$c_8$ \\
\hline
$x_{13,1}$ & $c_7$ & $x_{13,2}$& $c_8 $&$x_{13,3}$& $c_9$ \\
\hline
$x_{14,1}$& $c_8 $ & $x_{14,2}$& $c_9 $ &$x_{14,3}$& $c_{10} $ \\
\hline
$x_{15,1}$ & $c_9 $ & $x_{15,2}$& $c_{10} $ &$x_{15,3}$& $c_{11} $  \\
\hline
$x_{16,1}$ & $c_{10} $ & $x_{16,2}$& $c_{11} $ &$x_{16,3}$& $c_0$  \\
\hline
$x_{17,1}$ & $c_{11} $ & $x_{17,2}$& $c_0$ &$x_{17,3}$&$c_1$\\
\hline
\end{tabular}
\vspace{5pt}
\caption{Decoding of vector linear index code for SUICP(SNI) given in Example \ref{ex1}}
\label{table2}
\hrule
\end{table}
\end{example}

\begin{example}
\label{ex2}
Consider a SUICP(SCI) with $K=13,D=5,U=1$ and $m=1$. For this SUICP(SCI), we have $D+2m=7,U+2m=3$ and $a=2,b=3$ satisfy the equation $\text{gcd}(bK,b(D+2m+1)+a)\geq 6 \geq b(U+2m+1)$. The rate achieved by proposed construction is
\begin{align*}
1+\frac{D+m+\frac{a}{b}}{m+1}&=1+\frac{5+1+\frac{2}{3}}{2}\\& =\frac{26}{6}=3.66.
\end{align*}
We have $y_{k,j}=x_{k,2(j-1)+1}+x_{k-1,2(j-1)+2}$ for $k \in [0:12]$ and $j \in [1:3]$. The index code for this SUICP(SCI) is obtained by 
\begin{align*}
&[c_0~c_1~ \ldots ~c_{25}]\\&=[\underbrace{y_{0,1}~y_{0,2}~y_{0,3}}_{y_0}~\underbrace{y_{1,1}~y_{1,2}~y_{1,3}}_{y_1}~\ldots~
\underbrace{y_{12,1}~y_{12,2}~y_{12,3}}_{y_{12}}]\mathbf{L},
\end{align*}
where $\mathbf{L}$ is the AIR matrix of size $26 \times 39$ as given below.
\arraycolsep=0.5pt
\setlength\extrarowheight{-3.0pt}
{
$$\mathbf{L}_{39 \times 26}=\left[
\begin{array}{ccc}
\mathbf{I}_{26}\\
~~\\
\mathbf{I}_{13}~\mathbf{I}_{13} \\
 \end{array}
\right]$$
}

The 26 broadcast symbols for this SUICP(SCI) are given Table \ref{table3}. 

\begin{table*}[ht]
\centering
\setlength\extrarowheight{2pt}
\begin{tabular}{|c|c|}
\hline
$c_0=\underbrace{x_{0,1}+x_{12,2}}_{y_{0,1}}+\underbrace{x_{8,5}+x_{7,6}}_{y_{8,3}}$ & $c_1=\underbrace{x_{0,3}+x_{12,4}}_{y_{0,2}}+\underbrace{x_{9,1}+x_{8,2}}_{y_{9,1}}$\\
\hline
$c_2=\underbrace{x_{0,5}+x_{12,6}}_{y_{0,3}}+\underbrace{x_{9,3}+x_{8,4}}_{y_{9,2}}$ & $c_3=\underbrace{x_{1,1}+x_{0,2}}_{y_{1,1}}+\underbrace{x_{9,5}+x_{8,6}}_{y_{9,3}}$\\
\hline
$c_4=\underbrace{x_{1,3}+x_{0,4}}_{y_{1,2}}+\underbrace{x_{10,1}+x_{9,2}}_{y_{10,1}}$ & $c_5=\underbrace{x_{1,5}+x_{0,6}}_{y_{1,3}}+\underbrace{x_{10,3}+x_{9,4}}_{y_{10,2}}$\\
\hline
$c_6=\underbrace{x_{2,1}+x_{1,2}}_{y_{2,1}}+\underbrace{x_{10,5}+x_{9,6}}_{y_{10,3}}$ & $c_7=\underbrace{x_{2,3}+x_{1,4}}_{y_{2,2}}+\underbrace{x_{11,1}+x_{10,2}}_{y_{11,1}}$\\
\hline
$c_8=\underbrace{x_{2,5}+x_{1,6}}_{y_{2,3}}+\underbrace{x_{11,3}+x_{10,4}}_{y_{11,2}}$ & $c_9=\underbrace{x_{3,1}+x_{2,2}}_{y_{3,1}}+\underbrace{x_{11,5}+x_{10,6}}_{y_{11,3}}$\\
\hline
$c_{10}=\underbrace{x_{3,3}+x_{2,4}}_{y_{3,2}}+\underbrace{x_{12,1}+x_{11,2}}_{y_{12,1}}$ & $c_{11}=\underbrace{x_{3,5}+x_{2,6}}_{y_{3,3}}+\underbrace{x_{12,3}+x_{11,4}}_{y_{12,2}}$\\
\hline
$c_{12}=\underbrace{x_{4,1}+x_{3,2}}_{y_{4,1}}+\underbrace{x_{12,5}+x_{11,6}}_{y_{12,3}}$ & $c_{13}=\underbrace{x_{4,3}+x_{3,4}}_{y_{4,2}}+\underbrace{x_{8,5}+x_{7,6}}_{y_{8,3}}$\\
\hline
$c_{14}=\underbrace{x_{4,5}+x_{3,6}}_{y_{4,3}}+\underbrace{x_{9,1}+x_{8,2}}_{y_{9,1}}$ & $c_{15}=\underbrace{x_{5,1}+x_{4,2}}_{y_{5,1}}+\underbrace{x_{9,3}+x_{8,4}}_{y_{9,2}}$\\
\hline
$c_{16}=\underbrace{x_{5,3}+x_{4,4}}_{y_{5,2}}+\underbrace{x_{9,5}+x_{8,6}}_{y_{9,3}}$ & $c_{17}=\underbrace{x_{5,5}+x_{4,6}}_{y_{5,3}}+\underbrace{x_{10,1}+x_{9,2}}_{y_{10,1}}$\\
\hline
$c_{18}=\underbrace{x_{6,1}+x_{5,2}}_{y_{6,1}}+\underbrace{x_{10,3}+x_{9,4}}_{y_{10,2}}$ & $c_{19}=\underbrace{x_{6,3}+x_{5,4}}_{y_{6,2}}+\underbrace{x_{10,5}+x_{9,6}}_{y_{10,3}}$\\
\hline
$ c_{20}=\underbrace{x_{6,5}+x_{5,6}}_{y_{6,3}}+\underbrace{x_{11,1}+x_{10,2}}_{y_{11,1}}$  & $c_{21}=\underbrace{x_{7,1}+x_{6,2}}_{y_{7,1}}+\underbrace{x_{11,3}+x_{10,4}}_{y_{11,2}}$\\
\hline
$c_{22}=\underbrace{x_{7,3}+x_{6,4}}_{y_{7,2}}+\underbrace{x_{11,5}+x_{10,6}}_{y_{11,3}}$ & $c_{23}=\underbrace{x_{7,5}+x_{6,6}}_{y_{7,3}}+\underbrace{x_{12,1}+x_{11,2}}_{y_{12,1}}$\\
\hline
$c_{24}=\underbrace{x_{8,1}+x_{7,2}}_{y_{8,1}}+\underbrace{x_{12,3}+x_{11,4}}_{y_{12,2}}$  & $c_{25}=\underbrace{x_{8,3}+x_{7,4}}_{y_{8,2}}+\underbrace{x_{12,5}+x_{11,6}}_{y_{12,3}}$\\
\hline
\end{tabular}
\vspace{5pt}
\caption{Vector linear index code for SUICP(SNI) given in Example \ref{ex2}}
\label{table3}
\vspace{-5pt}
\end{table*}

Receiver $R_k$ required to decode five message symbols $x_{k,1},x_{k,2},\ldots,x_{k,6}$ for $k \in [0:12]$. Let $\tau_{k,j}$ be the code symbols used by receiver $R_k$ to decode $x_{k,j}$ for $k \in [0:12]$ and $j \in [1:6]$. Table \ref{table4} gives the code symbols used by each receiver to decode its wanted message symbol.

\begin{table*}[ht]
\centering
\setlength\extrarowheight{3pt}
\begin{tabular}{|c|c|c|c|c|c|c|c|c|c|c|c|}
\hline
$x_{k,1}$ &$\tau_{k,1}$&$x_{k,2}$&$\tau_{k,2}$&$x_{k,3}$&$\tau_{k,3}$ & $x_{k,4}$ &$\tau_{k,4}$&$x_{k,5}$&$\tau_{k,5}$&$x_{k,6}$&$\tau_{k,6}$ \\
\hline
$x_{0,1}$ & $c_0$ & $x_{0,2}$&$c_3$&$x_{0,3}$&$c_1$ & $x_{0,1}$ & $c_4$ & $x_{0,2}$&$c_2$&$x_{0,3}$&$c_5$\\
\hline
$x_{1,1}$ & $c_3$ & $x_{1,2}$&$c_6$&$x_{1,3}$&$c_4$ & $x_{1,1}$ & $c_7$ & $x_{1,2}$&$c_5$&$x_{1,3}$&$c_8$\\
\hline
$x_{2,1}$ & $c_6$ & $x_{2,2}$&$c_9$&$x_{2,3}$ &$c_7$ & $x_{2,1}$ & $c_{10}$ & $x_{2,2}$&$c_8$&$x_{2,3}$ &$c_{11}$\\
\hline
$x_{3,1}$ & $c_9$ &$x_{3,2}$&$c_{12}$&$x_{3,3}$ &$c_{10}$ &$x_{3,1}$ & $c_0+c_{13}$ &$x_{3,2}$&$c_{11}$&$x_{3,3}$ &$c_1+c_{14}$ \\
\hline
$x_{4,1}$ & $c_{12}$ &$x_{4,2}$&$c_2+c_{15}$&$x_{4,3}$&$c_0+c_{13}$ &$x_{4,1}$ & $c_3+c_{16}$ &$x_{4,2}$&$c_1+c_{14}$&$x_{4,3}$&$c_4+c_{17}$ \\
\hline
$x_{5,1}$ & $c_2+c_{15}$ & $x_{5,2}$&$c_5+c_{18}$ &$x_{5,3}$&$c_3+c_{16}$ & $x_{5,1}$ & $c_6+c_{19}$ & $x_{5,2}$&$c_4+c_{17}$ &$x_{5,3}$&$c_7+c_{20}$\\
\hline
$x_{6,1}$ & $c_5+c_{18}$ & $x_{6,2}$& $c_8+c_{21}$ &$x_{6,3}$& $c_6+c_{19}$ &$x_{6,1}$ & $c_9+c_{22}$ & $x_{6,2}$& $c_7+c_{20}$ &$x_{6,3}$& $c_{10}+c_{23}$ \\
\hline
$x_{7,1}$ & $c_8+c_{21}$ & $x_{7,2}$& $c_{11}+c_{24}$&$x_{7,3}$& $c_9+c_{22}$ & $x_{7,1}$ & $c_{12}+c_{25}$ & $x_{7,2}$& $c_{10}+c_{23} $&$x_{7,3}$& $c_{13}$\\
\hline
$x_{8,1}$& $c_{11}+c_{24} $ & $x_{8,2}$& $c_{14} $ &$x_{8,3}$& $c_{12}+c_{25} $ & $x_{8,1}$& $c_{15}$ & $x_{8,2}$& $c_{13}$ &$x_{8,3}$& $c_{16}$  \\
\hline
$x_{9,1}$ & $c_{14}$ & $x_{9,2}$& $c_{17} $ &$x_{9,3}$& $c_{15} $ & $x_{9,1}$ & $c_{18} $ & $x_{9,2}$& $c_{16} $ &$x_{9,3}$& $c_{19} $\\
\hline
$x_{10,1}$ & $c_{17} $ & $x_{10,2}$& $c_{20} $ &$x_{10,3}$& $c_{18}$ & $x_{10,1}$ & $c_{21} $ & $x_{10,2}$& $c_{19} $ &$x_{10,3}$& $c_{22}$ \\
\hline
$x_{11,1}$ & $c_{20} $ & $x_{11,2}$& $c_{23}$ &$x_{11,3}$&$c_{21}$ & $x_{11,1}$ & $c_{24} $ & $x_{11,2}$& $c_{22}$ &$x_{11,3}$&$c_{25}$\\
\hline
$x_{12,1}$ & $c_{23}$ & $x_{12,2}$&$c_0$&$x_{12,3}$&$c_{24}$ & $x_{12,1}$ & $c_1$ & $x_{12,2}$&$c_{25}$&$x_{12,3}$&$c_2$\\
\hline
\end{tabular}
\vspace{5pt}
\caption{Decoding of vector linear index code for SUICP(SNI) given in Example \ref{ex2}}
\label{table4}
\hrule
\end{table*}
\end{example}

\begin{example}
\label{ex3}
Consider a SUICP(SCI) with $K=71,D=17,U=3$ and $m=5$. For this SUICP(SCI), we have $D+2m=27,U+2m=13$ and $a=2,b=5$ satisfy the equation $\text{gcd}(bK,b(D+2m+1)+a) \geq b(U+2m+1)$. The rate achieved by proposed construction is
\begin{align*}
1+\frac{D+m+\frac{a}{b}}{m+1}&=1+\frac{17+5+\frac{2}{5}}{6}\\&=\frac{142}{30}=4.733.
\end{align*}
We have $y_{k,j}=x_{k,6(j-1)+1}+x_{k,6(j-1)+2}+x_{k,6(j-1)+3}+x_{k,6(j-1)+4}+x_{k,6(j-1)+5}+x_{k,6(j-1)+6}$ for $k \in [0:70]$ and $j \in [1:5]$. The $30$-dimensional vector linear index code for this SUICP(SCI) is obtained by 
\begin{align*}
&[c_0~c_1~ \ldots ~c_{141}]=\\&
[\underbrace{y_{0,1}~y_{0,2} \ldots y_{0,5}}_{y_0}~\underbrace{y_{1,1}~y_{1,2}\ldots y_{1,5}}_{y_1} \ldots \underbrace{y_{70,1}~y_{70,2}\ldots~y_{70,5}}_{y_{70}}]\mathbf{L},
\end{align*}
where $\mathbf{L}$ is the AIR matrix of size $355 \times 142$.

\end{example}
\subsection{Upperbound on the broadcast rate of SUICP(SCI)}
In \cite{VaR6}, we used extended Euclid algorithm to derive an upperbound on the broadcast rate of SUICP(SNI). In this subsection, we give an upperbound on the broadcast rate of SUICP(SCI) by using the results in \cite{VaR6}.
\begin{definition}
\label{def5}
Let $$a_{min}=\min_{(a,b) \in \mathbf{S}_{K,D,U,m}}~a.$$ Define $b_{min}$ as the corresponding value of $a_{min}$ such that $(a_{min},b_{min}) \in \mathbf{S}_{K,D,U,m}$.
\end{definition}

In \cite{VaR6}, we proved that $b_{min}$ is unique in a given $\mathbf{S}_{K,D,U,m}$ and $\frac{a_{min}}{b_{min}}$ is the value of $\frac{a}{b}$ such that $(a,b) \in \mathbf{S}_{K,D,U,m}$ and $\frac{a}{b}$ is minimum. In \cite{VaR6}, we gave an algorithm to find the values of $a_{min}$ and $b_{min}$ for the given SUICP(SNI). The algorithm given in \cite{VaR6} can be used to find the values of $a_{min}$ and $b_{min}$ for SUICP(SCI) by replacing $D$ with $D+2m$ and $U$ with $U+2m$.
\begin{theorem}
\label{upperbound}
Consider a SUICP(SCI) with arbitrary $K,D,U$ and $m$. The broadcast rate $\beta$ of this index coding problem is upperbounded by \begin{align*}
\beta \leq 1 + \frac{D+m+\frac{a_{min}}{b_{min}}}{m+1}.
\end{align*}
\end{theorem}
\begin{proof}
For every $(a,b) \in \mathbf{S}_{K,D,U,m}$, Theorem \ref{thm1} gives the construction of vector linear index codes for SUICP(SCI) with rate 
\begin{align*}
1 + \frac{D+m+\frac{a}{b}}{m+1}.
\end{align*}

We have $(a_{min},b_{min}) \in \mathbf{S}_{K,D,U,m}$. Hence, Theorem \ref{thm1} gives the construction of vector linear index codes for SUICP(SCI) with rate 
\begin{align*}
1 + \frac{D+m+\frac{a_{min}}{b_{min}}}{m+1}.
\end{align*}
The broadcast rate is the infimum of of all achievable rates. This completes the proof. 
\end{proof}
\begin{example}
\label{ex10}
Consider a SUICP(SCI) with $K=71,D=5,U=5$ and $m=2$. For this SUICP(SCI), we have $D+2m=21,U+2m=10$. By using the Algorithm given in \cite{VaR6},  we can obtain $a_{min}=1$ and $b_{min}=7$. Hence, the broadcast rate of this SUICP(SCI) is upperbounded by \begin{align*}
\beta \leq 1+\frac{D+m+\frac{a_{min}}{b_{min}}}{m+1}=1+\frac{5+2+\frac{1}{7}}{3}=\frac{22}{5}=3.38.
\end{align*}
\end{example}
\begin{note}
In Section \ref{sec2}, we shown that the broadcast rate of SUICP(SCI) is lowerbounded by 
\begin{align*}
\beta \geq 1+\frac{D+m}{m+1}=1+\frac{5+2}{3}=3.33.
\end{align*}
Hence, for this SUICP(SCI), we have  
\begin{align*}
3.33 \leq \beta \leq 3.38.
\end{align*}
\end{note}
\begin{example}
\label{ex11}
Consider a SUICP(SCI) with $K=93,D=11,U=1$ and $m=4$. For this SUICP(SCI), we have $D+2m=19,U+2m=9$. By using the Algorithm given in \cite{VaR6},  we can obtain $a_{min}=2$ and $b_{min}=3$. Hence, the broadcast rate of this SUICP(SCI) is upperbounded by \begin{align*}
\beta \leq 1+\frac{D+m+\frac{a_{min}}{b_{min}}}{m+1}=1+\frac{11+4+\frac{2}{3}}{5}=4.13.
\end{align*}
\end{example}
\begin{note}
In Section \ref{sec2}, we shown that the broadcast rate of SUICP(SCI) is lower bounded by 
\begin{align*}
\beta \geq 1+\frac{D+m}{m+1}=1+\frac{11+4}{5}=4.
\end{align*}
Hence, for this SUICP(SCI), we have
\begin{align*}
4.00 \leq \beta \leq 4.13.
\end{align*}
\end{note}

\section{Capacity of some SUICP(SCI)}
\label{sec5}
In this section, we derive the capacity of SUICP(SCI) for arbitrary $K,D$ and $U$, but $m$ satisfying the condition $\text{gcd}(K,D+2m+1)\geq U+2m+1$.

\begin{theorem}
\label{thm4}
Consider an SUICP(SCI) with arbitrary $K,D$ and $U$. For this SUICP(SCI), if $m$ satisfies the condition $\text{gcd}(K,D+2m+1)\geq U+2m+1$, then the capacity of this SUICP(SCI) is given by 
\begin{align*}
C=\frac{m+1}{D+2m+1}.
\end{align*}
\end{theorem}
\begin{proof}
For an SUICP(SCI) with arbitrary $K,D,U$ and $m$, in Theorem \ref{thm1}, we proved that we can combine $m+1$ message symbols into one extended message symbol and then we can use AIR matrix of size $Kb \times (b(D+2m+1)+a)$ to combine $Kb$ extended symbols into $b(D+2m+1)+a$ broadcast symbols to generate a $b(m+1)$ dimensional vector linear index code for the given SUICP(SCI). The rate achieved by using the proposed construction in Theorem \ref{thm1} is 
\begin{align}
\label{cap11}
\frac{b(D+2m+1)+a}{b(m+1)}=1+\frac{D+m+\frac{a}{b}}{m+1}.
\end{align}

Given
\begin{align*}
\text{gcd}(K,D+2m+1)) \geq U+2m+1.
\end{align*}

Hence, we have $(a,b)=(0,1) \in \mathbf{S}_{K,D,U,m}$.  With $a=0,b=1$, the rate achieved by the proposed construction given in Theorem \ref{thm1} reduces to 
\begin{align}
\label{cap12}
\frac{D+2m+1}{m+1}=1+\frac{D+m}{m+1}.
\end{align}
From Theorem \ref{lowerbound}, we have 
\begin{align}
\label{cap13}
\beta \geq 1+\frac{D+m}{m+1}.
\end{align}
From \eqref{cap12} and \eqref{cap13}, we have 
\begin{align}
\label{cap14}
\beta = 1+\frac{D+m}{m+1}=\frac{D+2m+1}{m+1}.
\end{align}
This completes the proof.
\end{proof}
\begin{corollary}
\label{cor5}
Consider an SUICP(SCI) with arbitrary $K,D$ and $m$. For this SUICP(SCI), if $U$ satisfies the condition $\text{gcd}(K,D+2m+1)\geq U+2m+1$, then the capacity of this SUICP(SCI) is given by 
\begin{align*}
C=\frac{m+1}{D+2m+1}.
\end{align*}
\end{corollary}
\begin{note}
In \cite{VaR4}, we derived the capacity of SUICP(SNI) for arbitrary $K$ and $D$, but $U=\text{gcd}(K,D+1)-1$. The capacity result in \cite{VaR4} is a special case of Theorem \ref{thm4} with $m=0$.
\end{note}
\begin{example}
\label{ex7}
Consider a SUICP(SCI) with $K=18,D=7,U=1$ and $m=2$. For this SUICP(SCI), we have $D+2m=11,U+2m=5$ and $\text{gcd}(K,D+2m+1)=\text{gcd}(18,12)=6 \geq U+2m+1$. The capacity of this SUICP(SCI) is
\begin{align*}
1+\frac{D+m}{m+1}=1+\frac{7+2}{3}=4.
\end{align*}
The optimal index code for this SUICP(SCI) is given in Example \ref{ex1}. 
\end{example}

\begin{example}
\label{ex8}
Consider a SUICP(SCI) with $K=55,D=13,U=2$ and $m=4$. For this SUICP(SCI), we have $D+2m=21,U+2m=10$ and $\text{gcd}(K,D+2m+1)=\text{gcd}(55,22)=11 \geq U+2m+1$. The capacity of this SUICP(SCI) is
\begin{align*}
1+\frac{D+m}{m+1}=1+\frac{13+4}{5}=\frac{22}{5}=4.4.
\end{align*}

The optimal index code for this SUICP(SCI) is a five dimensional vector linear index code. The index code for this SUICP(SCI) is given by 
\begin{align*}
&[c_0~c_1~ \ldots ~c_{21}]=[y_{0,1}~y_{1,1}~\ldots~y_{54,1}]\mathbf{L},
\end{align*}
where $\mathbf{L}$ is the AIR matrix of size $55 \times 22$ and $y_{k,1}=x_{k,1}+x_{k,2}+x_{k,3}+x_{k,4}+x_{k,5}$ for $k \in [0:54]$. 
\end{example}
\begin{theorem}
For arbitrary $K,D,U$ and $m$, the broadcast rate of SUICP(SCI) is bounded by 
\begin{align*}
1+\frac{D+m}{m+1} \leq \beta \leq 1+ \frac{D+m}{m+1}+ \frac{\frac{a_{min}}{b_{min}}}{m+1}.
\end{align*}
\end{theorem}
\begin{proof}
Proof follows from Theorem \ref{lowerbound} and \ref{upperbound}.
\end{proof}
\section{Conclusion}
\label{sec6}
In this paper, we constructed near optimal vector linear index codes for SUICP(SCI) with arbitrary $K,U,D$ and $m$. The constructed codes are independent of field size. We gave an upperbound and lowerbound on the broadcast rate of the SUICP(SCI) with arbitrary $K,U,D$ and $m$. We give the capacity of SUICP(SCI) with arbitrary $K,D$ and $U$ with $m$ satisfying the relation $\text{gcd}(K,D+2m+1) \geq U+2m+1$. The capacity and optimal index coding of SUICP(SCI) with arbitrary $K,D,U$ and $m$ is a challenging open problem.
\section*{Acknowledgment}
This work was supported partly by the Science and Engineering Research Board (SERB) of Department of Science and Technology (DST), Government of India, through J.C. Bose National Fellowship to B. Sundar Rajan.

\end{document}